\documentclass{article}[12pt]
\usepackage[utf8]{inputenc}
\usepackage[english]{babel}
\usepackage{amsmath, amsthm, accents, amssymb}
\newtheorem{remark}{Remark}
\newtheorem{theorem}{Theorem}
\newtheorem{lemma}{Lemma}[section]

\newtheorem{assumpA}{Assumption}

\DeclareMathOperator{\sign}{sign}
\usepackage{amsfonts,amssymb}
\usepackage{bm}
\usepackage{mathrsfs}
\usepackage{natbib}
\usepackage[title]{appendix}
\usepackage{diagbox}
\usepackage{caption}
\usepackage[margin=1.5in]{geometry}
\usepackage{graphicx}
\usepackage{setspace}
\usepackage{caption}
\usepackage{authblk}
\title{Stable combination tests}
\author{Xing Ling}
\author{Yeonwoo Rho\thanks{%Emails: X. Ling (xling@mtu.edu) and Y. Rho (yrho@mtu.edu) 
Corresponding author: Yeonwoo Rho, School of Mathematical Science, Michigan Technological University, Houghton, MI, 49931, USA. E-mail: yrho@mtu.edu. }}
\affil{Michigan Technological University}
\date{\today}

\begin{document}

\maketitle

\begin{abstract}
    This paper proposes a stable combination test, which is a natural extension of Cauchy combination tests by \cite{liu2020cauchy}. Similarly to the Cauchy combination test, our stable combination test is simple to compute, enjoys good sizes, and has asymptotically optimal powers even when the individual tests are not independent. This finding is supported both in theory and in finite samples. %In particular, the stable transformed p-values also follow a stable distribution asymptotically under the global null hypothesis. 
    \\
		\textit{key words and phrases}: Additive combination test; multiple hypothesis testing; stable distribution.
\end{abstract}

\section{Introduction}\label{section: introduction}

Multiplicity has been a long-standing issue in areas that often require testing a large number of hypotheses at the same time. Achieving higher power, while controlling the size, has been one of the most important missions in the multiplicity area. Dependency among the underlying individual tests is another key element to consider. For instance, the Bonferroni method can control the size even when the underlying tests are dependent. However, this method is so conservative that it has been criticized for its low power especially when the underlying tests are strongly positively correlated. See \cite{o1984procedures, moran2003arguments, dmitrienko2009multiple}, among others.

This paper focuses on additive p-value combination tests when the underlying p-values are allowed to be dependent. Additive combination tests utilize the fact that a p-value under a null hypothesis is uniformly distributed.
Individual p-values are transformed into new random variables, which are then linearly combined to form a combined test statistic. 
One of the key steps associated with a combined test statistic is figuring out its null distribution. This step is simpler if the transformation function is chosen in a way that the distribution of the transformed p-values under the null is closed under addition. For example, Fisher's combination test statistic is $-2\sum_{i=1}^n\log p_i$, where $p_i$s indicate the individual p-values. If the $i$th null hypothesis is true, $p_i$ would more or less follow a uniform distribution, and $-2\log p_i$ follows a gamma distribution with shape parameter 1 and scale parameter 2. If tests are independent, Fisher's statistic follows a gamma distribution with shape parameter $n$ and scale parameter 2. However, if tests are dependent, finding out the right shape and scale parameters is not easy. \cite{liu2020cauchy} found that a Cauchy distribution is a good alternative of the gamma distribution. Similarly to gamma random variables, a convex combination of standard Cauchy random variables is also standard Cauchy.
What makes a Cauchy distribution an attractive candidate for a combination test is that this relationship holds even when the random variables are dependent \citep{pillai2016unexpected}.
The Cauchy combination test (CCT) by \cite{liu2020cauchy} was born of this observation. The individual p-values are transformed to standard Cauchy random variables, and a convex combination of these Cauchy random variables is used as a combined test statistic. The critical values are taken from a standard Cauchy distribution. \cite{liu2020cauchy} proved that the CCT controls the size  well if the significance level is small and has asymptotically optimal powers under sparse alternatives. The CCT is fast to compute and robust to various forms of dependence structures.

Our paper is motivated by the fact that there is a wide class of distributions that is closed under addition -- strictly stable distributions. In fact, the Cauchy distribution is also a part of the strictly stable distribution family. This observation naturally leads to a stable combination test (SCT), which is a natural extension of the CCT. In the SCT, a stable distribution function is used in the transformation step. It is well-known that a linear combination of independent stably distributed random variables is still stable. In this paper, we show that the SCT statistic is also stably distributed asymptotically even when the underlying p-values are dependent, and therefore, can control the size successfully when the number of tests, $n$, is large enough. We also prove that the SCT has asymptotically optimal powers under sparse alternatives as long as $n$ is large enough. Our simulation results also suggest that our method is robust to dependent structures in finite samples.

This paper is organized as follows. In Section \ref{section: Stable framework}, we summarize the CCT and introduce our SCT method. The sizes and powers of the SCT are explored in theory.
In section \ref{section: simulation}, simulation results are provided to demonstrate favorable sizes and powers of the SCT in finite samples. Section \ref{section: conclusion} concludes. Technical proofs are relegated to the Appendix. 

Throughout this paper, we write $g(x)\sim h(x) $ as $x\to\infty$ to indicate $\lim_{x\to\infty} \frac{g(x)}{h(x)}=1$. The symbol $\mathbb{R}$ indicates the set of all real numbers. The symbol $A\setminus B$ indicates $A\cap B^C$. For instance, $\mathbb{R}\setminus \{0\}$ is the set of all real numbers except 0.  $\Gamma(\cdot)$ is the gamma function.
It is also worth mentioning that there are a few different ways to define stable distributions. In this paper, we follow \cite{nolan2020modeling}'s S1 parametrization.
According Definition 1.5 in \cite{nolan2020modeling}, a random variable $W$ is $\bm{S}(\alpha, \beta, \tau, \delta; 1)$ if $W$ has a characteristic function 
\begin{equation*}
    E[\exp{(iu W)}] = 
    \begin{cases}
    \exp\left\{-\tau^\alpha|u|^\alpha\left[1-i\beta\tan(\frac{\pi\alpha}{2})(\sign u)\right] + i\delta u\right\} & \alpha\neq1 \\
    \exp\left\{-\tau|u|\left[1+i\beta\frac{2}{\pi}(\sign u) \log|u|\right] + i\delta u\right\} & \alpha=1,
    \end{cases}
\end{equation*}
where $u\in (-\infty, \infty)$, $i = \sqrt{-1}$, and
$\sign u $ is the sign function which takes value $-1$ if $u<0$, $0$ if $u=0$, and $1$ if $u>0$.
There are four parameters: an index of stability $\alpha \in(0, 2]$, a skewness parameter $\beta\in[-1, 1]$, a scale parameter $\tau>0$, and a location parameter $\delta\in(-\infty, \infty)$. 
When the distribution is standardized, i.e., when scale $\tau=1$, and location $\delta=0$, the symbol $\bm{S}(\alpha, \beta)$ is an abbreviation for $\bm{S}(\alpha, \beta, 1, 0; 1)$. We write $F(x|\alpha, \beta, \tau, \delta) = \Pr (W<x)$, the distribution function of a stable random variable $W$ with parameters $\alpha, \beta, \tau, {\rm ~ and ~} \delta$. Similarly, when the distribution is standardized, $F(x|\alpha, \beta)$ is short for $F(x|\alpha, \beta, 1, 0)$.

\section{Stable Combination Test}\label{section: Stable framework}
 
Let $H_1, \ldots, H_n$ be individual null hypotheses and the corresponding alternative hypotheses be $H_1^c,\ldots,H_n^c$. The corresponding individual p-values are denoted by
$p_1, \ldots, p_n$. The global null hypothesis, defined as $\bm{H_0} = \bigcap_{i = 1}^n H_1$, is true when all $H_i$s are true. The global alternative hypothesis, defined as
$\bm{H_a} = \bigcup_{ i = 1}^n H_i^c$, is true if there is at least one false null hypothesis. We assume that all p-values are uniformly distributed under the global null hypothesis:
\begin{assumpA} \label{assump:uniform p}
Under the global null hypothesis $\bm{H_0}$,
$p_i$s are uniformly distributed for all $i=1, \ldots, n$ on $[0,1]$.
\end{assumpA}

\begin{remark}{\rm Assumption \ref{assump:uniform p} is satisfied if the individuals tests are exact tests with continuous test statistics. If individual tests are based on asymptotic results or on discrete statistics, the sample sizes should be large enough to satisfy this assumption approximately. 
However, this assumption can be relaxed to some non-uniform p-values as long as individual tests are more conservative than the nominal level. In this case, a combined p-value is also conservative, which means that the combined p-value can control the size correctly. See Remark 5 and Corollary 2 in \cite{liu2020cauchy} for more details. 
For brevity, the rest of this paper assumes uniform p-values under the null.
}\end{remark}

Inspired by \cite{pillai2016unexpected}'s finding that a sum of  dependent Cauchy random variables is still a Cauchy random variable, \cite{liu2020cauchy} proposed to combine p-values based on a Cauchy distribution. The p-values are first transformed into standard Cauchy random variables and then a weighted sum is taken. The test statistic is defined as the weighted sum
\begin{equation*}%\label{equation: CCT}
    T_{n, 1, 0}(\bm{p}) = \sum_{i  =1}^n {w}_i \tan [\pi(0.5-p_i)],
\end{equation*}
where ${w}_i$s are nonnegative weights and $\sum_{i=1}^n {w}_i = 1$. 
If the individual p-values are independent or perfectly dependent, it is straightforward that the test statistic follows the standard Cauchy distribution under the global null hypothesis. One of the main contributions of \cite{liu2020cauchy} is that the tail probability of the test statistic $T_{n;1, 0}(\bm{p})$ is approximately the same as that of a standard Cauchy distribution, even when $p_i$s are correlated. In order to prove the above statement, \cite{liu2020cauchy} assumed the p-values are calculated from standard normal distribution, $X_1, \ldots, X_n$, with $X_i$ follows $N(\mu_i, 1)$, $E[(X_1, \ldots, X_n)^T] = (\mu_1, \ldots, \mu_n)^T = \bm{\mu}$ and $Cov[(X_1, \ldots, X_n)^T] = \Sigma$. They also assume that $(X_i, X_j)$ for $i\neq j$ are pairwise bivariate normally distributed. 
In this case, $p_i=2[1-\Phi(|X_i|)]$ from the $i$th two-sided Z-test, and the test statistic can be rewritten as 
$$T_{n, 1, 0}(\bm{p})=T_{n, 1, 0}(\bm{X}) = \sum_{i = 1}^n {w}_i \tan\{\pi[2\Phi(|X_i|) - 1.5]\}.$$
They proved that $T_{n, 1, 0}(\bm{X})$ has the same tail probability as a standard Cauchy random variable if $\bm{\mu} = \bm{0}$. This means that the a standard Cauchy distribution can be used to derive the threshold for the combined p-value, $T_{n, 1, 0}(\bm{p})$.
\cite{liu2020cauchy} also proved that if $\bm{\mu}\neq\bm{0}$ and if $\bm{\mu}$ satisfies the sparse alternative assumption with large enough signals, this test has an asymptotically optimal power since $T_{n, 1, 0}(\bm{X}) \to\infty$ with probability 1.

Inspired by the fact that a Cauchy distribution is a special case of a stable distribution, we propose a Stable Combination Test (SCT). 
Let $W_{i;\alpha,\beta} = F^{-1}\left(1-p_i|\alpha, \beta\right)$ for $i = 1, \ldots, n$, where $F(\cdot|\alpha, \beta)$ is the distribution function of a standardized stable random variable with stability parameter $\alpha$ and skewness parameter $\beta$. The function $F^{-1}$ indicates the quantiles of $F$, defined by $F^{-1}(p|\alpha, \beta) = \inf\{x\in \mathbb{R}: F(x|\alpha, \beta)\ge p \}$. Though there are no closed forms for stable distribution function except for Normal, Cauchy, and L{\'e}vy distributions, stable quantiles can still be approximated numerically. We define our test statistic as follows:
\begin{equation}\label{equation: SCT statistic}
    T_{n;\alpha,\beta}(\bm{p}) = a_{n;\alpha} \sum_{i=1}^{n}w_i  W_{i;\alpha,\beta}, 
\end{equation}
where $w_i>0$ is the nonnegative weight imposed on $i$th test with $\sum_{i = 1}^n w_i = 1$ and $a_{n;\alpha} = \left(\sum_{j = 1}^n w_j^\alpha\right)^{-1/\alpha}$ is the normalizing factor.

We consider the stability parameters $0<\alpha< 2$. If $\alpha\neq1$, the skewness parameter ranges $-1<\beta\le 1$. If $\alpha =1$, only $\beta=0$ is considered. This is to ensure that $F(\cdot|\alpha,\beta)$ is strictly stable. A distribution is called strictly stable if the sum of i.i.d. random variables from this distribution follows the same distribution up to a normalizing factor without requiring a centering factor. Since $F^{-1}\left(1-p_i|\alpha, \beta\right)$ follows $\bm{S}(\alpha,\beta)$, $\sum_{i=1}^nw_iF^{-1}(1-p_i|\alpha,\beta)$ also follows $\bm{S}(\alpha,\beta)$ up to a normalizing factor if $\bm{S}(\alpha,\beta)$ is strictly stable. 
This motivates our definition of $T_{n;\alpha,\beta}(\bm{p})$ for $\alpha\neq1$ with the normalizing factor $\left(\sum_{j=1}^nw_j^\alpha\right)^{-1/\alpha}$. 
However, when $\alpha=1$, $\bm{S}(1,\beta)$ is no longer strictly stable unless $\beta = 0$.
When $\alpha=1$, and $\beta=0$, $\bm{S}(1,0)$ is a standard Cauchy distribution. Note that a naive extension of the CCT to different $\alpha$ and $\beta$, $\sum_{i = 1}^n w_i F^{-1}(1-p_i|\alpha, \beta)$, would not work without considering the normalizing factor $a_{n;\alpha}$.

\begin{remark}{\rm The test statistic $T_{n;1,\beta}(\bm{p})$ can still be defined for $\beta\neq0$ if an additional centering factor $\frac{2}{\pi}\beta\sum_{j = 1}^n w_j \log w_j$  is considered. However, this direction will not be elaborated in this paper for the following reasons: (i)  it had relatively poor sizes and powers in our unreported simulation, (ii) the requirements for the power proof need to be stronger if this case is included, and (iii) the computation for $F^{-1}(\cdot|1,\beta)$ is unstable if $\beta\neq0$. For these reasons, we only consider $\beta=0$ when $\alpha=1$ for the rest of this paper.
}\end{remark}

The rest of this section addresses that our SCT statistic (\ref{equation: SCT statistic}) is also approximately stably distributed under the global null hypothesis, even when the underlying p-values are not independent. This makes it possible to construct a test that can control the family-wise error rate. We also prove that our test has asymptotically optimal powers under alternatives.

\subsection{Size}

Under assumption \ref{assump:uniform p}, observe that $1-p_i$ is uniformly distributed under the global null hypothesis $\bm{H_0}$ for $1\le i \le n$, therefore, $W_{i;\alpha,\beta} = F^{-1}(1-p_i|\alpha, \beta)$ is identically distributed with marginal distribution $\bm{S}(\alpha, \beta)$. 
If the individual tests are independent, it is trivial that the test statistic 
\begin{equation}\label{distribution of test stat}
    T_{n;\alpha,\beta}(\bm{p}) \overset{d}{=} W_{0;\alpha,\beta},
\end{equation}
where $W_{0;\alpha,\beta}$ follows a stable distribution $\bm{S}(\alpha, \beta)$. This can be seen by simple computations using the property that the sums of $\alpha$-stable random variables are still $\alpha$-stable; see  Proposition 1.4 and equation (1.7) in \cite{nolan2020modeling}.

However, if $W_{i;\alpha,\beta}$ are not independent, there is no exact relationship as in (\ref{distribution of test stat}). Instead, an asymptotic relationship can be established when the number of tests, $n$, is large enough. For instance, \cite{jakubowski1989alpha} showed that a dependent sum of stable random variables is also asymptotically stable. In this paper, we adapt \cite{jakubowski1989alpha}'s Theorems 4.1 and 4.2 to establish that $T_{n;\alpha,\beta}(\bm{p})$ converges to $W_{0;\alpha, \beta}$ under some dependence assumptions in Theorem \ref{theorem: size} below. Thanks to this theorem, type I errors of SCTs can be controlled as long as $n$ is large enough. 

%In this paper, we adapt mixing and dependence assumptions from the extreme value theory field to model the dependence structure among the individual tests under the null. These assumptions aim to control long-range and short-range dependencies. 

The following Assumptions \ref{assump: long range dependence}, \ref{assump: rho mixing}, and \ref{assump: short range dependence} are adapted from equations (4.4), (4.8), and (4.5) of \cite{jakubowski1989alpha}, respectively.

\begin{assumpA}\label{assump: long range dependence}
Let $A\subset \mathbb{R}\setminus \{0\}$
%\mathbb R\cup \{-\infty, +\infty\} 
be a finite union of disjoint intervals of the form $(a, b]$ that do not contain 0, and  $A^c$ be the complementary set of $A$.
The sequence $\{W_{i;\alpha,\beta}\}_{i=1}^n$ satisfies 
\begin{equation*}
\begin{split}
    \sup_{1\le p<q<r\le n}& \Bigg| \Pr\left(\bigcap_{p\le i\le r} ( a_{n;\alpha} w_iW_{i;\alpha,\beta}\in A^c)\right) - \\
    &\Pr\left(\bigcap_{p\le i\le q} ( a_{n;\alpha} w_iW_{i;\alpha,\beta}\in A^c)\right)
    \Pr\left(\bigcap_{q\le i\le r} ( a_{n;\alpha} w_iW_{i;\alpha,\beta}\in A^c)\right) \Bigg| \to 0
\end{split}
\end{equation*}
for every $A$ as $n\to\infty$. 
\end{assumpA}

\begin{assumpA} \label{assump: rho mixing}
The sequence $\{W_{i;\alpha,\beta}\}_{i=1}^n$ is $\rho$-mixing with $\sum_{j =1}^\infty \rho(2^j) < +\infty$. 
A sequence $\{X_i\}_{i=1}^n$ is called $\rho$-mixing if
\begin{equation*}
        \rho(m) = \sup_{1\le i\le j\le n} \sup \left\{|corr(f, g)|: f\in \mathscr{L}^2(\mathscr{F}_i^j), g\in \mathscr{L}^2(\mathscr{F}_{j+m}^\infty)\right\} \xrightarrow[m\to\infty]{}0,
\end{equation*}
where $\mathscr{F}_{i}^j$ is the $\sigma$-field generated by $(X_{i}, \ldots, X_{j})$  and $\mathscr{L}^2 (\mathscr{F}_{i}^j)$ be the space of square-integrable, $\mathscr{F}_{i}^j$-measurable random variables. 
\end{assumpA}

\begin{assumpA}\label{assump: short range dependence}
Let $\Delta(r)$ be an arbitrary division of the set $\{1, 2, \ldots, n\}$ into $r$ segments, $0 = m_0\le m_1\le \cdots\le m_r = n$. For every $\varepsilon >0$, the sequence $\{W_{i;\alpha,\beta}\}_{i=1}^n$ satisfies
\begin{equation*}
    \lim_{r\to \infty} \limsup_{n\to\infty} \inf_{\Delta(r)} \sum_{q = 1}^r \sum_{m_{q-1}<i<j\le m_q}^{} \Pr\left(a_{n;\alpha} w_i|W_{i;\alpha,\beta}|
    >\varepsilon, a_{n;\alpha} w_j|W_{j;\alpha,\beta}|>\varepsilon\right)=0.
\end{equation*}
\end{assumpA}

The first two assumptions, Assumptions \ref{assump: long range dependence} and \ref{assump: rho mixing}, mainly concern the long-range dependence. Assumption \ref{assump: long range dependence} basically assumes asymptotic long-range independence and will be used to address the convergence in distribution of our test statistic with $0<\alpha<1$. Assumption \ref{assump: rho mixing} is a $\rho$-mixing condition, which will be used for $1\le \alpha<2$.
Assumption \ref{assump: short range dependence} limits the amount of short-range dependence by assuming that large values cannot be clustered in a small segment \citep{beirlant2006statistics}.  In our setting under the global null, Assumption \ref{assump: short range dependence} means that at most one $W_i$ can have large absolute value within a small neighborhood.
If $W_i$s are independent, this condition can be easily satisfied. However, this condition may not hold if the short range dependence is too strong.

\iffalse
Assumption \ref{assump: short range dependence} was originally introduced by \cite{leadbetter1974extreme} as their condition $D'$. \cite{davis1983limit, davis1983stable} established a bivariate dependence condition $D'$, which considers not only a large exceedance but also a small one.
%$$\lim \limsup n\sum P(X_1>\varepsilon, X_j>\varepsilon) = 0$$
\cite{husler1986extreme} weaken the Leadbetter's condition $D'$ by dividing the set $\{1, 2, \ldots, n\}$ into segments. Moreover, it has advantages for the verification of the conditions \citep{husler1986extreme}.
%$$\lim \max_{I} \min_{I*\subset I} \sum_{i<j\in I*}P(X_i>\varepsilon, X_j>\varepsilon) = 0$$
We adopt the condition $D_0'$ from \cite{jakubowski1989alpha}, who
gave a general expression of \cite{husler1986extreme}'s condition $D'$. Assumption \ref{assump: short range dependence}, or equation (4.5) of \cite{jakubowski1989alpha}, is weaker than condition D' of \cite{davis1983stable}. However, this does not mean that our set of assumptions is weaker than that of \cite{davis1983stable}, since our Assumption \ref{assump: long range dependence} is much stronger than Condition D of \cite{davis1983stable}.
\fi

\begin{remark}{\rm
Note that long-range and short-range dependencies make the best sense either when there is a natural order among the individual tests or when the tests are independent. This situation is not too unusual in practice. For instance, any sequential testing, including testing a sequence of genes, would fall within this category. %Network data may also be able to satisfy Assumptions \ref{assump: long range dependence}-\ref{assump: short range dependence}, even if they do not have a natural order, if their connection is sparse enough.
}\end{remark}

Suppose \ref{assump: long range dependence} and \ref{assump: short range dependence} or  \ref{assump: rho mixing} and \ref{assump: short range dependence}  are satisfied. Define an i.i.d. sequence $\{\tilde{W}_{i;\alpha,\beta}\}$ that has the same marginal distribution as $\{W_{i;\alpha,\beta}\}$. Note that
$\tilde{T}_{n;\alpha,\beta}(\bm{p}) = a_{n;\alpha}\sum_{i = 1}^n w_i \tilde{W}_{i;\alpha,\beta}$ follows a
$\bm{S}(\alpha, \beta)$ distribution for any $n$, using a similar argument as in  (\ref{distribution of test stat}). By 
applying Theorems 4.1 and 4.2 of \cite{jakubowski1989alpha}, our test statistic  converges in distribution to $\bm{S}(\alpha, \beta)$.
This observation leads to the following Theorem \ref{theorem: size}.

\begin{theorem}\label{theorem: size}
Let $W_{0;\alpha,\beta}$ be a random variable that follows $\bm{S}(\alpha, \beta)$, where $0<\alpha<2$ and $-1\le \beta\le  1$. Assume one of the following conditions:
\begin{enumerate}
    \item[(i)] $0<\alpha<1$ and Assumptions \ref{assump: long range dependence} and  \ref{assump: short range dependence} hold.
    \item[(ii)] $1\le\alpha<2$ and Assumptions  \ref{assump: rho mixing} and \ref{assump: short range dependence} hold.
\end{enumerate}
Then, if the global null hypothesis $\bm{H_0}$ is true,
\begin{equation*}%\label{convergence in distribution}
    T_{n;\alpha,\beta}(\bm{p}) = a_{n;\alpha} \sum_{i = 1}^n w_i W_{i;\alpha,\beta}\overset{d}{\to} \bm{S}(\alpha, \beta)
\end{equation*}
as $n\to\infty$.
\end{theorem}

Based on Theorem \ref{theorem: size}, the global null hypothesis is rejected at significance level $s$ if $T_{n;\alpha,\beta}(\bm{p})>t_s$, where $t_s$ is the upper $s$ quantile of $\bm{S}(\alpha,\beta)$.

%\begin{remark}{\rm
%It is worth noting that when $0<\alpha<1$ and $\beta=1$, $F^{-1}$ is always positive. Using a similar argument as in \cite{wilson2019harmonic}, we can show that the SCT controls the FWER in the strong sense. However, when $\alpha\geq1$, the same argument cannot be used any more.}\end{remark}

\begin{remark}{\rm
A stable distribution with parameters $\alpha=1$ and $\beta =0$ is a Cauchy distribution. In this case, our SCT is equivalent to CCT \citep{liu2020cauchy}; i.e.
$$T_{n;1,0}(\bm{p}) = \sum_{i =1}^n w_i \tan[\pi(0.5 - p_i)].$$
\cite{liu2020cauchy}'s method is robust to dependencies among the underlying p-values, similarly to ours. While our theorems for the SCT do cover include the CCT, our technical settings are slightly different from those of   \cite{liu2020cauchy}.
The first difference lies in the forms of dependencies allowed in assumptions. Assumption C.1 of \cite{liu2020cauchy} assumed that every pair of test statistics of individual tests is bivariate normal. While the p-values follow a uniform distribution marginally, their pairwise dependencies are modelled  through bivariate Gaussian copula. On the contrary, our assumptions do not require Gaussian copulas. Instead, we control long-range and short-range dependencies. This means that our assumptions require a structure in dependence such as a natural order. Our assumptions also requires relatively weaker dependencies, whereas \cite{liu2020cauchy} did not impose any restrictions on the strength of dependencies. The second difference is that \cite{liu2020cauchy}'s test is controlled only when the significance level $s$ is small enough, while ours would work for any $s$. This is because \cite{liu2020cauchy}'s Theorems 1 and 2 concern the right tail probabilities only. They showed that the right tail probability of the test statistic is approximately the same as the right tail probability of a Cauchy random variable only when the significance level $s$ is small enough. By contrast, the type I error of the SCT can be controlled at any significance level as long as $n$ is large enough. This is because the in distribution convergence result in our Theorem \ref{theorem: size} is much stronger than the right tail convergence results in Theorems 1 and 2 in \cite{liu2020cauchy}. 
The last difference is that our result holds only when the number of individual tests, $n$, is large enough, while \cite{liu2020cauchy}'s Theorem 1 showed that the CCT can control the size also when $n$ is fixed.

\iffalse
The above mentioned differences are mostly technical. Since the result of our theorem is stronger than that of \cite{liu2020cauchy}, covering a wider range of $\alpha$s and guaranteeing in distribution convergence rather than the right tail approximation, it is natural that our assumptions are stronger. In particular, we need to assume more structure in dependence such as natural orders, and the strength of dependence should be weaker than that is allowed in \cite{liu2020cauchy}. We believe that the CCT and other SCTs are not 
\fi
\iffalse
However, when tests are strongly dependent, there is an interesting difference between between the SCTs with $\alpha\neq1$ and the SCT with $\alpha=1$ (CCT). For instance, if $p_i$ are perfectly correlated, the CCT statistic $T_{n;1,0}(\bm{p})$ is exactly standard Cauchy. On the contrary, the other SCTs do not follow the same distribution. In this case, $T_{n;\alppha,\beta}=a_{n;\alpha}W_{1;\alpha,\beta}$ follows $\bm{S}(\alpha,\beta,a_n,0)$. Note that $a_n=1$ if $\alpha=1$, $a_n\to 0$ if $\alpha<1$, and $a_n\to\infty$ if $\alpha>1$ as $n\to\infty$. It seems that the CCT is the only case where the limiting distribution can maintain the standard form despite the perfect correlation. Other tests will need to change their scale parameter $\tau$ in the perfect correlation case. In particular, if $\alpha<1$, the limiting distribution becomes degenerate.
\fi
}\end{remark}

\begin{remark}\label{remark: Stouffer's Z-score}{\rm 
The form of our SCT statistic also resembles Stouffer's Z-score \citep{stouffer1949american}.
A stable distribution with tail parameter $\alpha =2$ is a normal distribution no matter what the skewness parameter $\beta$ is. In this case, our SCT test statistic is equivalent to Stouffer's Z-score; i.e.
$$T_{n;2,\beta}(\bm{p}) = \frac{1}{\sqrt{\sum_{j=1}^nw_j^2}} \sum_{i=1}^n w_i\Phi^{-1}(1-p_i).$$
Stouffer's Z-score \citep{stouffer1949american, mosteller1954selected} method was designed for independent hypotheses. \cite{abelson2012statistics} found that Stouffer's test is more sensitive to consistent departures from the null hypothesis than Fisher's method for independent tests. Although \cite{kim2013stouffer} found that Stouffer's test works well in the analysis of large scale microarray data for dependent tests, and the form of our SCT statistic can cover Stouffer's Z-score, we do not include $\alpha=2$ in our proof for Theorem \ref{theorem: size}. This is because the simulation results in Section \ref{section: simulation} show that Stouffer's Z-score always performs worse than the SCTs with $\alpha<2$. In particular, Stouffer's Z-score tends to severely over-reject under strong dependencies. Their size-adjusted powers are always dominated by the other choices of $\alpha$s. Accordingly, even though it is not impossible that Stouffer's Z-score still works in dependent cases, we do not pursue this direction in theory.
}\end{remark}

\subsection{Power}\label{section: power}

In this section, we prove that our SCT test statistics have asymptotically optimal powers under sparse alternative hypotheses. We consider a similar setting to the one in \cite{liu2020cauchy}. 
Let $\bm{X} = [X_1, X_2, \ldots, X_n]^T$ be the collection of test statistics, where $X_i$ corresponds to the $i$-th individual test. Suppose $X_i$ marginally follows a normal distribution. Denote $E[\bm{X}] = \bm{\mu}$ and $Cov(\bm{X}) = \bm{\Sigma}$. Without loss of generality, we assume $\bm{\Sigma}$ is the correlation matrix, i.e., each $X_i$ has variance 1. The p-value for $i$-th two-sided test is $p_i = p(X_i) = 2[1-\Phi(|X_i|)]$. The global null hypothesis is  specified as $\bm{H_0}: \bm{\mu} = \bm{0}$ versus the global alternative hypothesis $\bm{H_a}: \bm{\mu} \neq \bm{0}$.

\begin{assumpA} \label{assump: alternative}
Let $S = \{1\le i\le n: \mu_i \neq 0\} $, the collection of indices for which the individual null hypotheses $H_i$s are false. Let $S_{+} =\{1\le i\le n: \mu_i > 0\}$, and assume $|S_{+}|\ge|S|/2$ without loss of generality.  
\begin{enumerate}
    \item\label{assump part: S complementary} The p-values in $S^c$ follow a uniform distribution and $\{W_{i;\alpha,\beta}\}_{i\in S^c}$ satisfy the requirements in Theorem \ref{theorem: size} with $a_{|S^c|;\alpha}$. 
    \item\label{assump part: sparse} The number of elements in $S$ is $ n^{\gamma}$ with $0<\gamma<0.5$. 
    \item\label{assump part: magnitude} The magnitude for all nonzero $\mu_i$ is the same. For $i\in S$, $|\mu_i| = \mu_0= \sqrt{2r\log n}$, and $\sqrt{r} + \sqrt{\gamma} >\max\{\sqrt{\alpha},1\}$.
    \item \label{assump part: min weight} There exists a positive constant $c_0$ such that $\min_{i=1}^n w_i \ge c_0 n^{-1}$. The sum of weights $\sum_{j\in S}w_j = n^{\gamma-1}$.
 
\end{enumerate}
\end{assumpA}

Part \ref{assump part: S complementary} of Assumption \ref{assump: alternative} requires the p-values in the set $S^c$ satisfy Assumptions \ref{assump: long range dependence} and \ref{assump: short range dependence} if $0<\alpha<1$ or satisfy Assumption  \ref{assump: rho mixing} and \ref{assump: short range dependence} if $1\le \alpha<2$. Under this condition, the contribution of p-values in the set $S^c$ to the test statistic is bounded. Part \ref{assump part: sparse} requires a sparse alternative, which is commonly taken in the multiple testing field. Part \ref{assump part: magnitude} controls the strength of signals. The magnitude of the nonzero signals should be large enough to ensure the test statistic is arbitrary large. Part \ref{assump part: min weight} helps keep the contribution of $\max_{i\in S}p_i$ under control. Note that $p_i$ can still be close to 1 even when $i\in S$. In this case, $F^{-1}(1-p_i)$ can be negative, possibly leading to a less powerful test. This assumption is to guarantee that such $p_i$s would not affect the power of the test asymptotically.

\begin{theorem}\label{theorem: power}
Consider $0<\alpha<2$ and $-1<\beta\le 1$. Under Assumption \ref{assump: alternative}, for any significance level $s$, the power of the SCT converges to 1 as $n\to\infty$:
$$\lim_{n\to\infty}\Pr[T_{n;\alpha,\beta}(\bm{p})>t_s] =1,$$
where $t_s$ is the upper $s$-quantile of stable distribution $\bm{S}(\alpha, \beta)$, i.e., $F(t_s|\alpha, \beta) =1-s$. 
\end{theorem}

The proof is attached in Appendix \ref{append: proof on power}. 

For Theorem \ref{theorem: power}, we no longer consider $\beta=-1$. The powers of small $\beta$s tend to be dominated by other $\beta$s given the same $\alpha$, making it not worth considering $\beta=-1$ for a powerful test. This is because the left tail becomes heavier as $\beta$ gets closer to $-1$, which prevents a test from having better powers. See Section \ref{section: simulation} for a related discussion.

\section{Simulation Results}\label{section: simulation}
In this section, we explore the size, raw power and size-adjusted power of the SCT in finite samples in a similar setting as \cite{liu2020cauchy}. A collection of test scores, $ \bm{X}$, is drawn from $N_{n}(\bm{\mu}, \bm{\Sigma})$. All diagonal elements of the covariance matrix $\bm{\Sigma}$ are set as 1. There are four models for the covariance matrix $\bm{\Sigma}$ considered to represent different dependence structures. 
Model 1 is the scheme where the individual tests are independent. 
In Models 2, 3, and 4, the off-diagonal entries of the covariance matrix $\bm{\Sigma} = (\sigma_{ij})$ are functions of $\rho$.
\begin{enumerate}
    \item Independent. The correlation between each pair of underlying test scores is zero, i.e., $\bm{\Sigma} = I_n$. 
    \item AR(1) correlation. The correlation between a pair of underlying test scores  decays exponentially fast as their distances increase;
    $\sigma_{ij} = \rho^{|i-j|}$. 
    \item Exchangeable structure. The correlation between each pair of underlying test scores $\sigma_{ij} = \rho$ for all $i \neq j$. 
    \item Polynomial decay. The correlation between the $i$th and $j$th test scores, $\sigma_{ij}$, is set to be $\frac{1}{0.7 + |i-j|^\rho}$. It should be noted that the correlation is a decreasing function of $\rho$, unlike Models 2 and 3 above.
    
\end{enumerate}

The simulation is conducted in \texttt{R}. We use \texttt{qstable} function in \texttt{stabledist} package \citep{wuertz2016package} to calculate quantiles of stable distributions. We truncate too small and too large p-values at $10^{-6}$ and $1-10^{-6}$, respectively. This is to avoid technical issues involved with too large quantiles in absolute values in the \texttt{qstable} function.  The number of Monte Carlo replications is 1000. The number of  individual tests in each Monte Carlo replication is 40 ($n = 40$). The significance level is set to be $5\%$. The parameter $\rho$ that governs the strength of the dependencies is set to be $0.2, 0.4, 0.6,$ or $0.8$. 
Note that larger $\rho$ implies stronger dependencies in Models 2 and 3, and weaker dependencies in Model 4. For the SCT, all combinations of $\alpha = 0.1, 0.3,\ldots, 1.9$ and $\beta = -0.8,-0.6,\ldots, 1$ are considered in addition to $(\alpha,\beta)=(1,0)$, which is equivalent to the CCT. We also consider Stouffer's Z-score, which would correspond to the SCT with $\alpha=2$ and $\beta=0$. Note that although Stouffer's Z-score can be written in the SCT form, Stouffer's Z-score is not a part of the SCT family we consider in our paper. The test statistics are calculated as equation (\ref{equation: SCT statistic}) with equal weights $w_i = 1/n$.

When calculating the sizes, data are generated from a multivariate normal distribution with mean $\bm{\mu} = \bm{0}$. For powers, a sparse alternative hypothesis with $\gamma = 0.43$ and $r = 0.54$ is considered. This means that we set $\mu_i=\sqrt{2r\log n} \approx 2$ to randomly chosen $[n^{\gamma}]=[n^{0.43}] = 4$ indices in each replication. 
Note that $(\sqrt{0.43} + \sqrt{0.54} )^2\approx 1.93>\max(\alpha,1)$ for all $\alpha$s considered in our simulation, satisfying Part \ref{assump part: magnitude} of Assumption \ref{assump: alternative}. For raw powers, the cutoff values are taken directly from the corresponding stable distributions.
For the size-adjusted powers,  1000 Monte Carol replications are first drawn under the global null hypothesis. Combined test statistics are calculated for each Monte Carlo replication. The simulation-based cutoff for each method is determined as the 95\% quantile of the 1000 test statistics. After that, another set of 1000 Monte Carlo replications is drawn under the sparse alternative. The number of test statistics that are greater than the simulation-based cutoffs are counted to compute the size-adjusted powers. 

Figures \ref{figure: indep}-\ref{figure: Poly alpha>2} present the sizes, raw powers, and size-adjusted powers for different $\alpha$s and $\beta$s under the four models. Red dots indicate the CCT, which corresponds to the SCT with $\alpha=1$ and $\beta=0$. Black dots indicate Stouffer's Z-scores. The black solid lines in the size plots represent the nominal significance level $0.05$.

Figure \ref{figure: indep} presents the sizes, raw powers, and size-adjusted powers when the underlying tests are independent. All methods considered in our simulations, including the Stouffer's Z-score, are supposed to work fine in this case. The Stouffer's Z-score and the SCT with $\alpha=1.7$ and $\beta=-0.8$ have the best size, 0.05. However, these two methods are not necessarily the best due to their relatively low powers. In particular, the Stouffer's Z-score is the lowest in both raw and size-adjusted powers. The SCT with $\alpha=1.7$ and $\beta=-0.8$ also has relatively low raw and size-adjusted powers. In this independent scenario, the SCT with $\alpha\ge1$ and $\beta\ge0$ tends to have higher powers without loosing the size control.
In particular, most SCT methods including the CCT controls the size quite successfully, although there is a slight tendency of under-rejection. 
In terms of powers of the SCT, when $\alpha<1$, $\beta$ does not seem to affect the powers much, whereas when $\alpha>1$ the powers have an increasing trend as $\beta$ increases.
It is noticeable the CCT tends to have better powers than the SCTs with $\alpha<1$, keeping the sizes similar.
However, when $\alpha>1$ and $\beta>0$, the SCT performs better in general than the CCT, both in sizes and powers.
In particular, when $\alpha = 1.5, 1.7$ and $\beta = 1$, the SCT performs best with well-controlled sizes and highest size-adjusted powers. 
%In particular, when $\alpha = 1.3, 1.5, 1.7$ and $\beta \ge 0$, the SCT has higher raw and size-adjusted powers than CCT. 
%When $\alpha$ is close to 1 or 2 (i.e., $\alpha = 1.1, 1.9$), $\beta$s close to 1 can beat the CCT by a small margin. 

% stouffer of three models
Models 2-4 represent dependent cases. In these cases, the SCT works better than Stouffer's Z-score. When tests are dependent, Stouffer's Z-score is not supported in our theorems. See Remark \ref{remark: Stouffer's Z-score}. Stouffer's Z-score is the weakest in our simulations as well, as can be seen in Figures \ref{figure: AR alpha>2}-\ref{figure: Poly alpha>2}. In Models 2-4, Stouffer's Z-score tends to over-reject, and this tendency gets worse as the dependency gets stronger. Stouffer's Z-score also tends to have much lower powers than the SCT methods. While it sometimes has decent raw powers (e.g., Model 3 with $\rho=0.6$ and $0.8$), these powers are inflated due to their higher sizes. Their size-adjusted powers are consistently low in settings.

As for the behavior of the SCT family including the CCT, it seems that different sets of $\alpha$ and $\beta$ work better in different situations. There is a tendency that the stronger the dependency is, the more oversized larger $\alpha$s and the more undersized smaller $\alpha$s. Smaller $\alpha$s are generally required to keep the sizes under control for stronger dependencies. However, too small $\alpha$s may result in too conservative tests. In general, the SCTs with $\alpha\ge1$ paired with larger $\beta$s tend to have well-controlled sizes and better powers for models with weaker or no dependencies, whereas $\alpha\le 1$ paired with larger $\beta$s tend to have better performances when dependencies are stronger.

Figure \ref{figure: AR alpha>2} presents the sizes and powers in Model 2. With weaker dependencies ($\rho=0.2$ or $0.4$), the SCT with $\alpha\ge 1$ has the best size-adjusted powers and well-controlled sizes. 
In particular, when $\rho=0.2$, SCT with $\alpha = 1.7$ and $\beta =1$ works the best with sizes less than 0.05 and highest raw powers and size-adjusted powers.
When $\rho = 0.4$, the SCT with $\alpha = 1.3$ and $\beta = 0.6$ works the best with size very close to the target value and the highest size-adjusted power.
In Model 2 with higher dependencies ($\rho=0.6$ or $0.8$), the SCT with $\alpha< 1$ has the best size-adjusted powers and well-controlled sizes. 
For instance, when $\rho = 0.6$, SCTs with $\alpha = 0.9$ have the highest size-adjusted powers and under-controlled sizes. 
When $\rho = 0.8$, the SCT with $\alpha = 0.1$ has the highest size-adjusted power with the size under control. The effect of $\beta$ is not as much. In general, $\beta=1$ works reasonably well for all $\alpha<1$.

Figure \ref{figure: Exch alpha>2} presents the sizes and powers of Model 3. The dependencies in Model 3 are stronger than those of Model 2 given the same $\rho$. As a result, smaller $\alpha$s tend to work better in this case compared to the Model 2 cases.
%which is the reason why there are more over-rejections with large $\alpha$s. When $\alpha$ is close to 0, under-rejection tends to get worse as the dependency gets stronger. When $\alpha$ is close to 2, over-rejection tends to get worse with higher dependencies. 
In Model 3, when dependency is relatively weak with $\rho=0.2$, the SCT with $\alpha = 1.1$ and $\beta = 1$ works best with size 0.048, raw powers 0.459 and size-adjusted powers 0.469. 
When the dependency is moderate or strong in Model 3, the SCT with $\alpha$ close to 0 and $\beta$ close to 1 tends to have the best size-adjusted powers. In particular, when $\rho = 0.4, 0.6,$ or $0.8$, the SCT with $\alpha = 0.1$ and $\beta = 1$ has the greatest size-adjusted powers and controlled sizes.

Figure \ref{figure: Poly alpha>2} presents model 4 cases. Model 4's dependencies are even stronger than those of Model 3 in general. Therefore, smaller $\alpha$s, compared to other models, tend to produce better sizes and powers. Note that unlike Models 2 and 3, the larger $\rho$ is, the weaker the dependencies are in this case. For all four $\rho$s considered, $\alpha<1$ and larger $\beta$s tend to control sizes better with higher power. In particular, when $\rho = 0.4, 0.6, 0.8$ in Model 4, $\alpha = 0.5$ or $0.7$ and $\beta = 1$ have the best size-adjusted powers and under-controlled sizes. Under the strongest dependency setting with $\rho = 0.2$, the SCTs with $\alpha = 0.1, 0.3$ and $\beta = 1$ produce the highest size-adjusted powers. 
In terms of raw powers when $\rho = 0.2, 0.4$ or 0.8, the SCTs with $\alpha = 0.9$ and $\beta = 1$ work the best with under-controlled sizes and largest raw powers. When $\rho = 0.6$, the SCT with $\alpha = 0.9$ and $\beta = 0.8$ has the highest raw power.

The different performances of different $\alpha$s and $\beta$s seem to be strongly connected to how heavily the tails of the transformed p-values are. The right tail probability of a stable random variable is $P(W_{0;\alpha,\beta}>x)\sim c_\alpha(1+\beta)x^{-\alpha}$ as $x\to\infty$, where $c_\alpha=\sin(\pi\alpha/2)\Gamma(\alpha)/\pi$ and $W_{0;\alpha,\beta}$ is a stable random variable that follows $\bm{S}(\alpha,\beta)$. This right tail approximation holds for all $0<\alpha<2$ and $-1<\beta\le 1$. It is noteworthy that the right tail probability is an increasing function in $\beta$ and a decreasing function in $\alpha$ for large enough $x$. This means that the smaller the $\alpha$ is and the larger the $\beta$ is, the stable transformed p-values $W_{i;\alpha,\beta}$ used for the our combined test statistic in equation (\ref{equation: SCT statistic}) has heavier right tails. It seems that this heavier right tail is particularly useful when the dependence is strong for the size control. One interesting observation is that the heavier tail seems to affect in different ways under the null and the alternative. Under the null, the heavier tails lead to reject less, often correcting the over-rejection behavior under stronger dependencies. This is also the cause of under-rejection when $\alpha$ is too small. Under the alternative, the effects of heavier tail vary depending on the source. The heavier tail induced by larger $\beta$s usually lead to reject more, leading to better raw powers. On the contrary, the heavier tail behavior due to smaller $\alpha$s on powers is not monotone. The raw powers increase as $\alpha$ increases, with a peak at around $\alpha= 1.5$ or $1.7$, and then rapidly decreases. The only exception to the above observation on powers is Model 3 with $\rho=0.8$. In this case, the powers decrease as $\beta$ increases when $\alpha>1$, and the powers tend to increase as $\alpha$ increases.

The power behaviors due to $\alpha$s are somewhat consistent with the size behavior. However, the increasing powers as $\beta$ increases cannot be explained by the heavier right tails. This behavior as well as the under-rejection for very small $\alpha$s may be explained by the left tail behavior. Our SCT is an additive combination test where its summands may be negative. When the p-values are too close to 1, the stable transformed p-values take large negative values, which might reduce the chance of detecting the false null hypothesis when added to the test statistic. The large p-values are connected to the left tail probability of a stable distribution, which approximately has a power law $\Pr(W_{0;\alpha,\beta}<-x)\sim c_\alpha(1-\beta)x^{-\alpha}$ for large positive $x$ when $-1\le\beta<1$. When $\beta=1$, 
$\Pr(W_{0;\alpha,1}<-x)<P(W_{0;\alpha,\beta}<-x)$ for any $\beta<1$. This means that the left tail probability of $W_{0;\alpha,\beta}$ is a decreasing function in $\beta$ for all $-1\le\beta\le1$ as well as in $\alpha$, unlike the right tail probability. As a result, the smaller $\alpha$s and $\beta$s are, the heavier the left tails are. Since heavier left tails may result in the loss of powers, the powers could decrease as $\alpha$ and $\beta$ decreases. This is indeed consistent with our observations in powers in most cases.

In addition, notice that the effect of $\alpha$ on both tails is exponential whereas that of $\beta$ is only linear. In particular, for small $\alpha$s, the effect of $\beta$ is not as noticeable. This is because the effect of $\alpha$ dominates over the effect of $\beta$ in these cases. On the contrary, the effect of $\beta$s is more noticeable both in sizes and powers when $\alpha$ is relatively large. This is because the changes in the tail probabilities due to $\alpha$ is dominated by that of $\beta$ for $\alpha$s closer to 2.

It is also noteworthy that the SCT's behavior is somewhat consistent in all models. In particular, the exchangeable structure in Model 3 does not satisfy our long-range independence assumption in Assumptions \ref{assump: long range dependence} and \ref{assump: rho mixing}, unlike the other two dependent models, Models 2 and 4. The fact that the SCT's finite sample behavior in Model 3 was similar to that of in Models 2 and 4 implies that the SCT can in fact be applied to a wider range of conditions.

In summary, the SCT can control the sizes in finite samples for all the four models when $0<\alpha\le 1$ even under strong dependencies, unlike the Stouffer's Z-score. However, when $1<\alpha< 2$, sizes tend to be substantially inflated under moderate and strong dependencies in finite samples. The size behaviors can be explained by how heavy the right tails of the transforming stable distributions are. In general, the heavier right tails seem to help control the size against strong dependencies. The heavier right tails are obtained when $\alpha$ is small and $\beta$ is larger. This explains why smaller $\alpha$ and larger $\beta$s are preferred in the strong dependency case.

The powers of the SCTs tend to decrease as the dependency gets stronger. In general, the SCTs with $\alpha>1$ and large $\beta$ tend to have the best powers under no or weak dependencies, whereas the SCTs with $\alpha\le 1$ and large $\beta$ have the best sizes and powers under moderate and strong dependencies. The powers are affected by how heavy left tails of the transforming stable distributions. In general, larger $\alpha$s and $\beta$s lead to lighter left tails, which often result in better powers.

Based on this simulation results, we recommend using the SCT with $\alpha \approx 1.5$ and $\beta=1$ if the dependence is suspected to be relatively low, and using the SCT with $0.5\le\alpha\le1$ and $\beta\ge0$ when the individual tests are suspected to be strongly dependent. If there is no knowledge on the strength of the dependencies, we recommend using either the CCT or the SCT with $\alpha\approx 0.9$ and $\beta=1$, which lead to the best size-controlled tests without loosing too much power in most cases in our simulation.

\begin{figure}
    \centering
   \includegraphics[width = 13.9cm,height = 4.6cm]{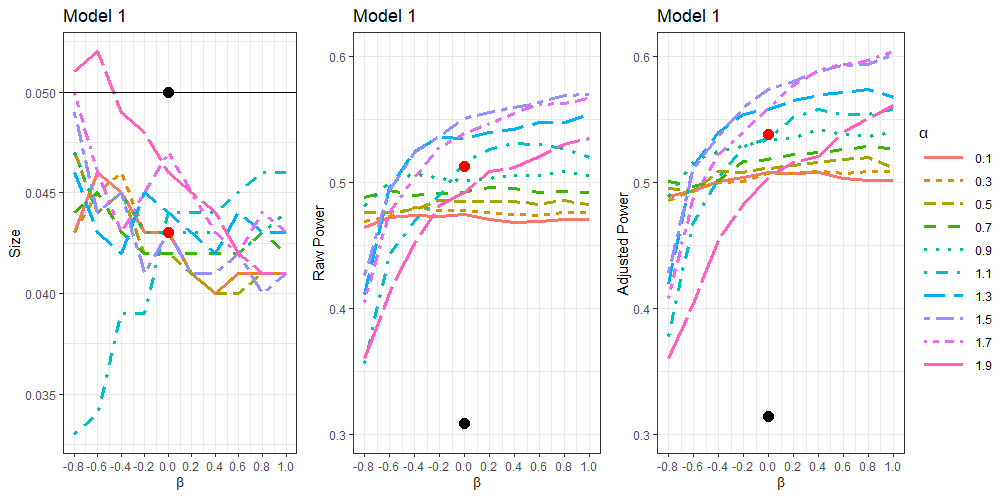}
    \caption{Sizes, raw powers and size-adjusted powers of Model 1 where tests are independent. Lines indicate the SCT with different $\alpha$s and $\beta$s. Red and black dots represent the CCT (SCT with $\alpha=1$ and $\beta=0$) and Stouffer's Z-score, respectively.  }
    \label{figure: indep}
\end{figure}

\begin{figure}
    \centering
\includegraphics[width = 13.9cm,height = 4.6cm]{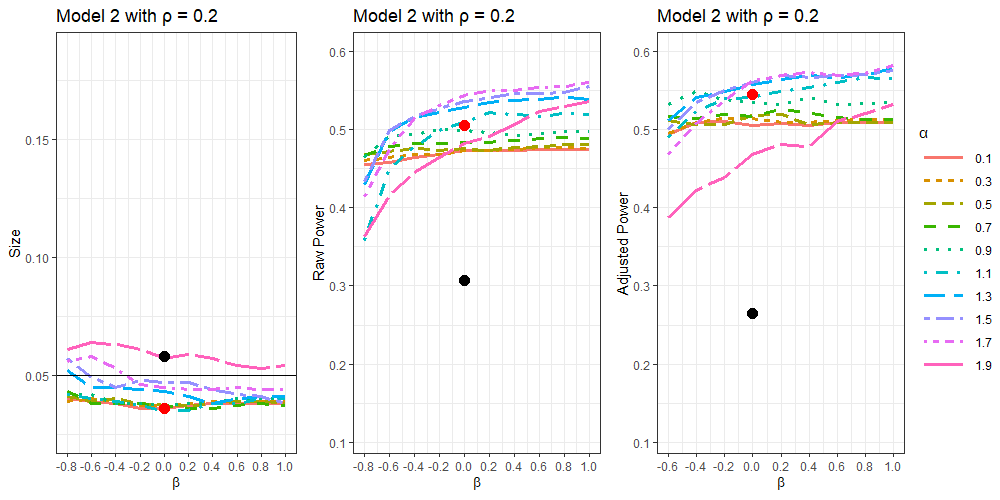}
\includegraphics[width = 13.9cm,height = 4.6cm]{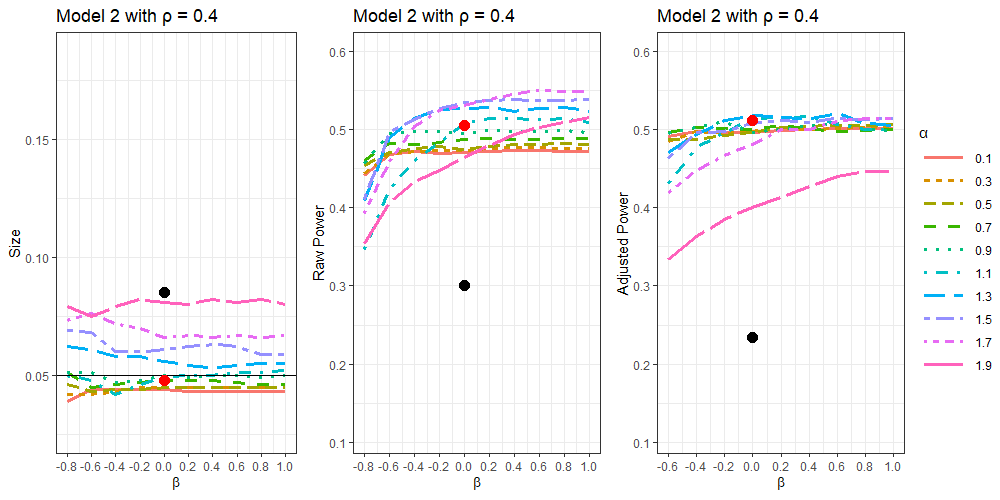}
\includegraphics[width = 13.9cm,height = 4.6cm]{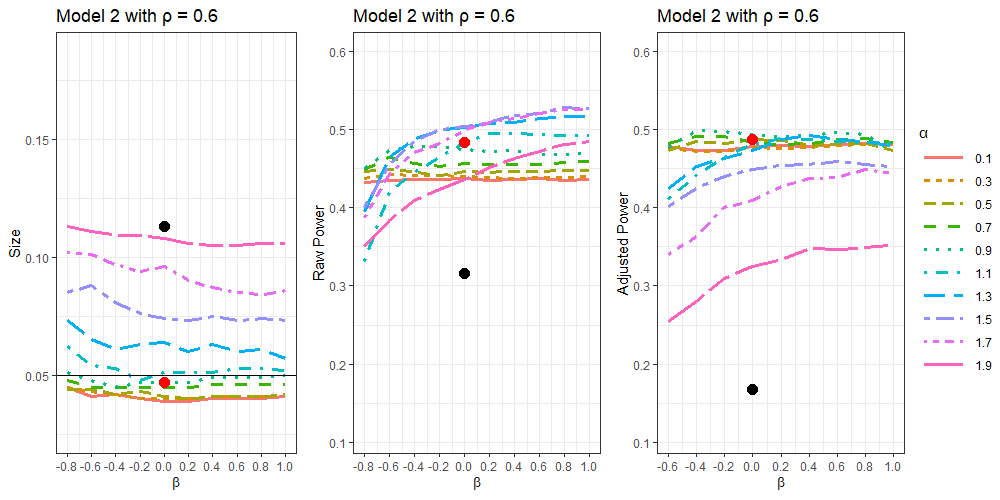}
\includegraphics[width = 13.9cm,height = 4.6cm]{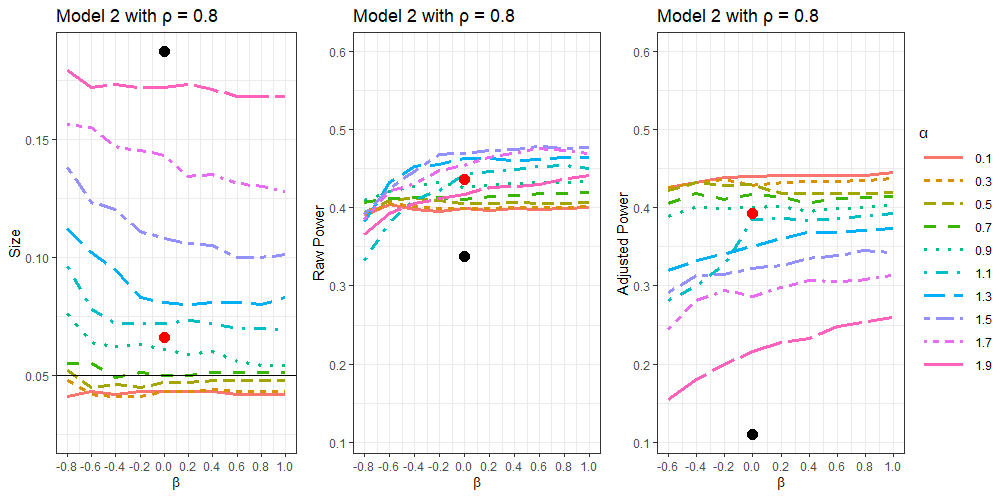}
    \caption{Sizes, raw powers and size-adjusted powers of Model 2 where tests are correlated with AR(1) correlation structures with different $\rho$s. Lines indicate the SCT with different $\alpha$s and $\beta$s. Red and black dots represent the CCT (SCT with $\alpha=1$ and $\beta=0$) and Stouffer's Z-score, respectively.  }
    \label{figure: AR alpha>2}
\end{figure}

\begin{figure}
    \centering
    \includegraphics[width = 13.9cm,height = 4.6cm]{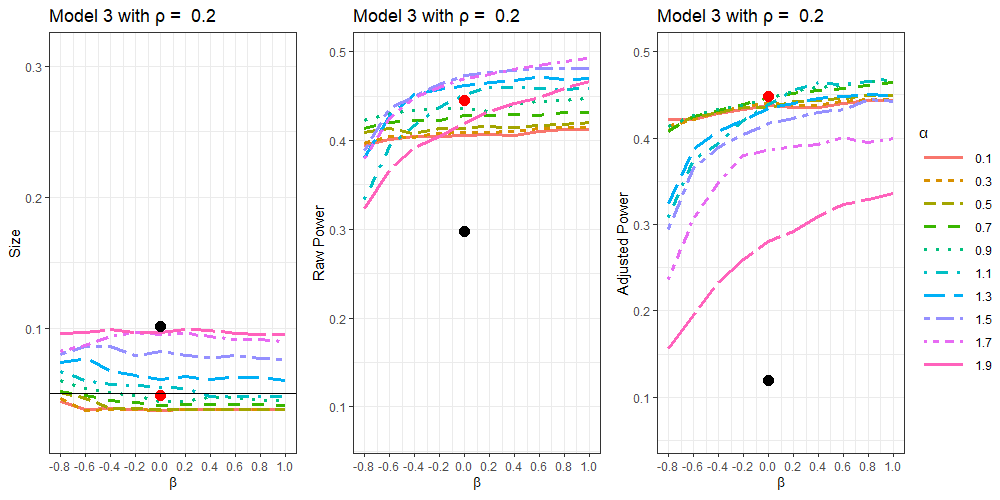}
\includegraphics[width = 13.9cm,height = 4.6cm]{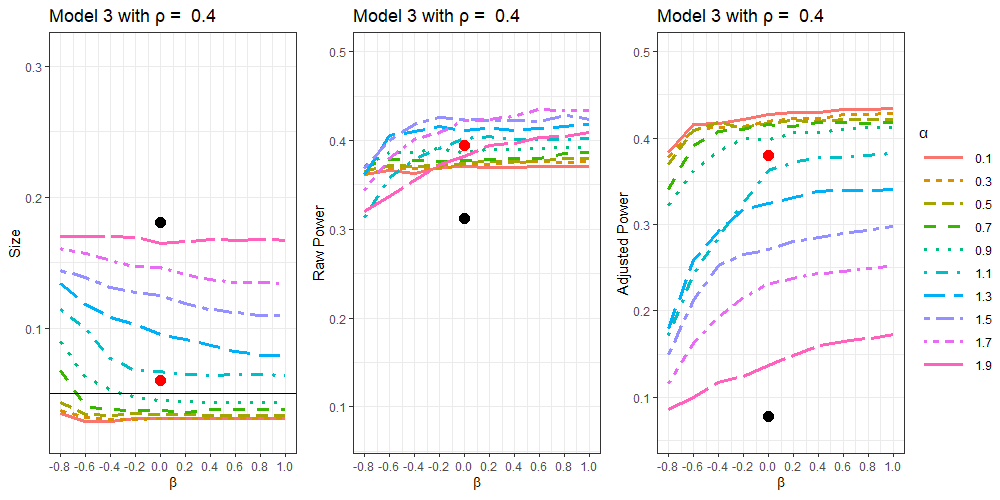}
\includegraphics[width = 13.9cm,height = 4.6cm]{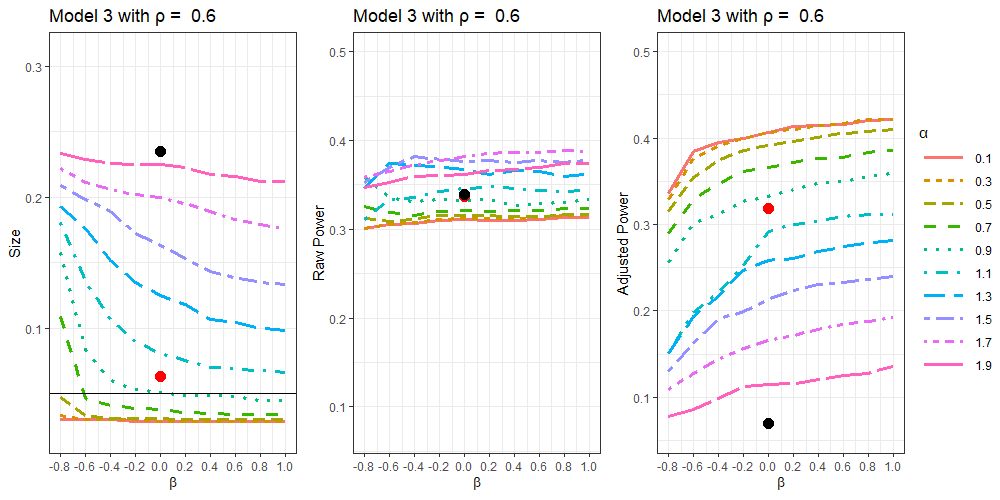}
\includegraphics[width = 13.9cm,height = 4.6cm]{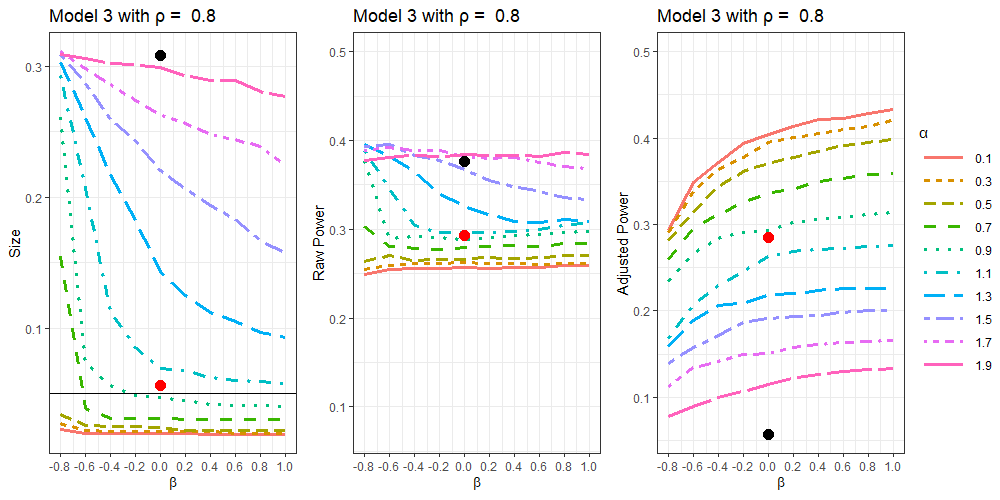}
    \caption{Sizes, raw powers and size-adjusted powers of Model 3 where tests are correlated with exchangeable correlation structures with different $\rho$s. Lines indicate the SCT with different $\alpha$s and $\beta$s. Red and black dots represent the CCT (SCT with $\alpha=1$ and $\beta=0$) and Stouffer's Z-score, respectively.  }
    \label{figure: Exch alpha>2}
\end{figure}

\begin{figure}
    \centering
    \includegraphics[width = 13.9cm,height = 4.6cm]{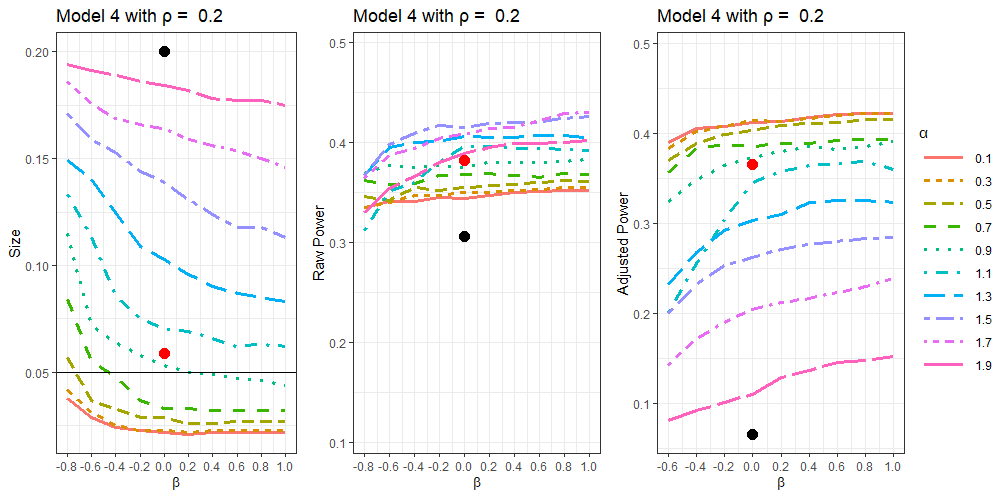}
    \includegraphics[width = 13.9cm,height = 4.6cm]{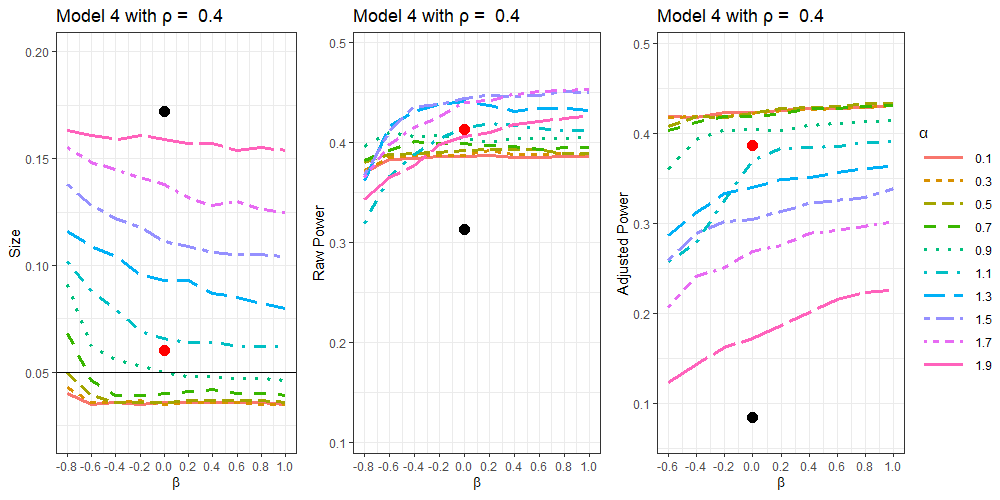}
\includegraphics[width = 13.9cm,height = 4.6cm]{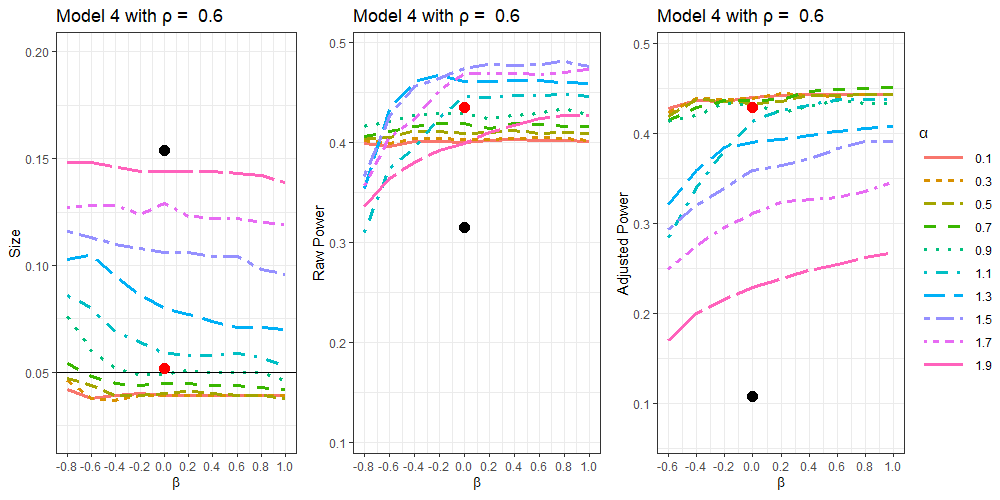}
\includegraphics[width = 13.9cm,height = 4.6cm]{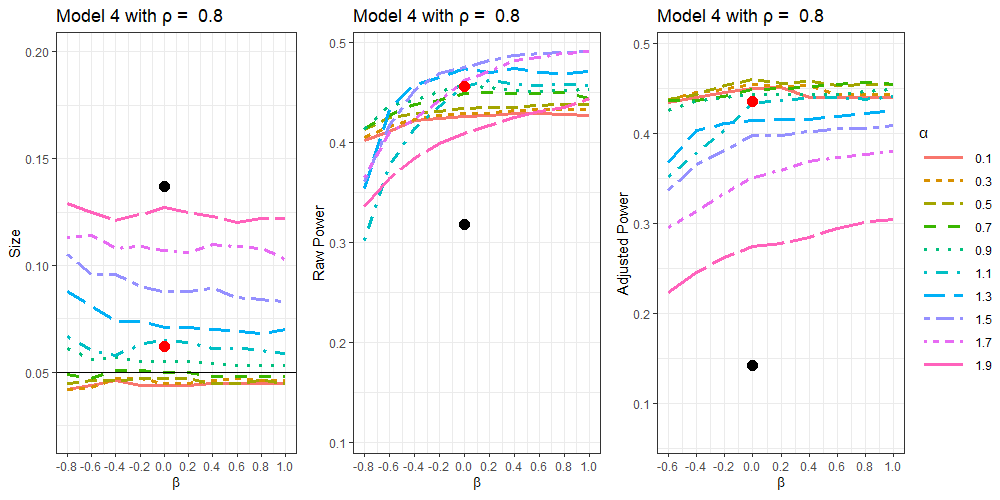}
    \caption{Sizes, raw powers and size-adjusted powers of Model 4 where tests are correlated as polynomial decayed correlation structures with different $\rho$s. Lines indicate the SCT with different $\alpha$s and $\beta$s. Red and black dots represent the CCT (SCT with $\alpha=1$ and $\beta=0$) and Stouffer's Z-score, respectively.  }
    \label{figure: Poly alpha>2}
\end{figure}

\section{Conclusion}\label{section: conclusion}

In this paper, we formulated an additive combination test based on stable distributions. The individual p-values are first transformed into a stably distributed random variables and then their weighted sum is considered. This weighted sum still has a stable distribution, making it possible to construct a test for the global null hypothesis. This method can be considered as an extension of the Cauchy combination test, which is based only on the Cauchy distributed random variables, because Cauchy distribution is also a stable distribution. Similarly to \cite{liu2020cauchy}'s result, our test is robust to  some forms of dependencies among individual p-values. We proved that this new test can successfully control the size and has asymptotically optimal power, which is further confirmed in simulations.

\iffalse
[[comment out from here]]
We provide a method for nonasymptotically multiple hypotheses testing by transform the individual p-values into a stable distributed random variables. We proved that the weighted sum of these stable random variables still has a stable distribution, which fact can be adopted to test the global null hypothesis. 

There are several limitations in the current study. First, we have a dependence assumption like condition \ref{assump: indep tail prob} rather than arbitrary dependence structure.  
Second, the individual p-values are require to be uniformly distributed. However, many test is asymptotic. 

Possible future directions are to relax the dependence assumption and provide theoretical works on the choice of function $\phi$, which makes the Kolmogorov's generalized mean has an asymptotically power law behaviour at infinity. 

\fi

\section*{Acknowledgement}
Rho and Ling are partially supported by NSF-CPS grant \#1739422.

\begin{appendices}

\pagebreak
\renewcommand{\theequation}{A.\arabic{equation}}
\setcounter{equation}{0}

\section{Technical Lemmas}

This section presents lemmas for the proof of Theorem \ref{theorem: power}. Recall that $p(x)=2-2\Phi(|x|)$ for all $x\in \mathbb{R}$, as defined in the beginning of Section \ref{section: power}. Lemmas \ref{lemma: large x} and \ref{lemma: small x} help find the lower bound of $F^{-1}(1-\min_{i\in S} p_i)$ and $F^{-1}(1-\max_{i\in S} p_i)$, respectively. Lemma \ref{lemma:an} presents a lower bound for $a_{n;\alpha}$.

\begin{lemma} \label{lemma: large x}
Define $g(x) = c_{\alpha, \beta}x^{1/\alpha} e^{\frac{x^2}{2\alpha}} $
with constant $c_{\alpha, \beta} = \left[\frac{1+\beta}{\sqrt{2\pi}}\Gamma(\alpha) \sin\left(\frac{\pi\alpha}{2}\right) \right]^{1/\alpha}, $ 
%(\sqrt{\frac{\pi}{2}}b_1)^{1/\alpha} 
where $0<\alpha<2$ and $-1<\beta\le 1$. For $x\to\infty$, %For large enough $x$ with order $O(\sqrt{log(n)})$, 
\begin{equation*}
    F^{-1}[1-p(x)|\alpha, \beta] > g(x) = c_{\alpha, \beta}x^{1/\alpha} e^{\frac{x^2}{2\alpha}} .
\end{equation*}
\end{lemma}

\begin{proof}[Proof of Lemma \ref{lemma: large x}]

When $x\to\infty$, $g(x)\to\infty$. Therefore, we can apply the right tail approximation of a stable distribution in Theorem 1.2 form \cite{nolan2020modeling}. When $0<\alpha<2$ and $-1< \beta\le1$, 
\begin{equation*}
\begin{split}
     1- F[g(x)|\alpha, \beta] &=\Pr[W_{0;\alpha, \beta}>g(x)]\\
     & \sim \frac{1+\beta}{\pi}\Gamma(\alpha)\sin\left(\frac{\pi\alpha}{2}\right) [g(x)]^{-\alpha}\\
     &=\sqrt{\frac{2}{\pi}}x^{-1}e^{-x^2/2}.
\end{split}
\end{equation*}
From Mill's ratio inequality that $1-\Phi(x)\le \frac{\phi(x)}{x}$ for any $x>0$, where $\Phi(\cdot)$ and $\phi(\cdot)$ represent the distribution function and probability density function of a standard normal random variable respectively, we have
\begin{equation*}
    \begin{split}
       p(x) & = 2[1-\Phi(x)] \\
       &\le 2\frac{\phi(x)}{x}\\
       & =\sqrt{\frac{2}{\pi}} x^{-1}e^{-x^2/2}.\\
    \end{split}
\end{equation*}
Therefore, $p(x) \le 1- F[g(x)|\alpha, \beta]$ for $x\to\infty$. Since $F^{-1}$ is increasing,  $F^{-1}[1-p(x)]>g(x)$ for large enough $x$. 

\end{proof}

\begin{lemma}\label{lemma: small x}

Define $\tilde{g}(x) = -\tilde{c}_{\alpha, \beta}x^{-1/\alpha}e^{\frac{x^2}{2\alpha}}$ with constant $\tilde{c}_{\alpha, \beta} = \left[\frac{1-\beta}{\sqrt{2\pi}}\Gamma(\alpha) \sin\left(\frac{\pi\alpha}{2}\right) \right]^{1/\alpha}  $ 
where $0<\alpha<2$ and $-1\le \beta\le  1$. When $x\to0^{+}$,
\begin{equation*}
    F^{-1}[1-p(x)|\alpha, \beta] > \tilde{g}(x) = - \tilde{c}_{\alpha, \beta}x^{-1/\alpha} e^{\frac{x^2}{2\alpha}}.
\end{equation*}

\end{lemma}

\begin{proof}[Proof of Lemma \ref{lemma: small x}]
We first consider the case where $-1\le\beta<1$.
Similarly to the proof of Lemma \ref{lemma: large x}, when $x\to0^+$, $\tilde{g}(x)\to -\infty$, and thus we can apply the left tail approximation from Theorem 1.2 of \cite{nolan2020modeling} when $0<\alpha<2$ and $-1\le \beta<1$: 
\begin{equation*}
\begin{split}
     F[\tilde{g}(x)|\alpha, \beta] &=\Pr[W_{0;\alpha, \beta}<\tilde{g}(x)]\\
     & \sim \frac{1-\beta}{\pi}\Gamma(\alpha)\sin\left(\frac{\pi\alpha}{2}\right) [-\tilde{g}(x)]^{-\alpha}\\
     &=\sqrt{\frac{2}{\pi}}x e^{-x^2/2}.
\end{split}
\end{equation*}
The standard normal distribution function, $\Phi(x)$, can be rewritten with integration by parts, 
\begin{equation*}
    \begin{split}
        1-p(x) &= 2\Phi(x) - 1 \\
        & = \sqrt{\frac{2}{\pi}}x e^{-x^2/2}  + Q (x),
    \end{split}
\end{equation*}
where $Q (x) = \sqrt{\frac{2}{\pi}} e^{-x^2/2} (\frac{x^3}{3}+\frac{x^5}{3*5} + \cdots)>0$ if $x>0$.
Therefore, $ 1-p(x) > F[\tilde{g}(x)|\alpha, \beta]$ for $x\to 0^+$.  

When $\beta=1$, the distribution is totally skewed to the right, and the left tail probability does not follow a power law. Instead, we know that the left tail probability of $W_{0;\alpha, 1}$ is smaller than that of $W_{0;\alpha, \beta}$ with $-1\le \beta<1$. That is,
$$F[\tilde{g}(x)|\alpha, 1] =\Pr[W_{0;\alpha, 1}<\tilde{g}(x)] \le \Pr[W_{0;\alpha, \beta}<\tilde{g}(x)]. $$
Therefore, 
$ 1-p(x) > F[\tilde{g}(x)|\alpha, \beta] > F[\tilde{g}(x)|\alpha, 1]$ for all $-1\le \beta \le 1$. Since $F^{-1}$ is increasing, we have  $F^{-1}[1-p(x)]>\tilde{g}(x)$, which completes the proof.
\end{proof}

\begin{lemma}\label{lemma:an} Let $w_i\in(0, 1)$ be nonnegative weights such that $\sum_{i=1}^n w_i =1$. The normalizing constant $a_{n;\alpha} = \left(\sum_{i=1}^n w_i^\alpha\right)^{-1/\alpha}\ge \min\{n^{1-1/\alpha}, 1\}$.
\end{lemma}
\begin{proof}[Proof of Lemma \ref{lemma:an}]
The lower bound of $a_{n;\alpha}$ is considered in three separate cases. 
First, when $\alpha=1$, $a_{n;\alpha}=1$. The second case is when $0< \alpha< 1$. By Hölder's inequality, 
$$ \sum_{i=1}^n w_i^\alpha \le \left[\sum_{i=1}^n\left( w_i^\alpha\right)^{1/\alpha}\right]^\alpha n^{1-\alpha} = n^{1-\alpha},$$
which is equivalent to 
$ a_{n;\alpha} \ge n^{1-1/\alpha}.$
The last case is when $1<\alpha<2$. From the fact that $l_\alpha$ norm is smaller $l_1$ norm, 
$$\left(\sum_{i=1}^n |w_i|^\alpha \right)^{1/\alpha}\le \sum_{i=1}^n |w_i| = 1,$$
and therefore, $a_{n;\alpha}\ge 1 .$
Combining the above three cases, we have $a_{n;\alpha} = \left(\sum_{i=1}^n w_i^\alpha\right)^{-1/\alpha}\ge \min\{n^{1-1/\alpha}, 1\}$. 

\end{proof}

\section{Proof of Theorem \ref{theorem: power}}\label{append: proof on power}
\begin{proof}[Proof of Theorem \ref{theorem: power}]
Recall that the test statistic is defined as $T_{n;\alpha,\beta}(\bm{p}) = T_{n;\alpha,\beta}(\bm{X}) = a_{n;\alpha} \sum_{i=1}^{n}w_i F^{-1}[1-p(X_i)|\alpha, \beta]$, where  $a_{n;\alpha} =  \left(\sum_{j = 1}^n w_j^\alpha\right)^{-1/\alpha} $. 
Under Assumption \ref{assump: alternative}, 
the test statistic $T_{n;\alpha,\beta}(\bm{X})$ can be decomposed into two parts:
$$T_{n;\alpha,\beta}(\bm{X}) = a_{n;\alpha} \sum_{i\in S}w_i F^{-1}[1-p(X_i)|\alpha, \beta] +a_{n;\alpha} \sum_{i\in S^c}w_i F^{-1}[1-p(X_i)|\alpha, \beta]:= A_n+B_n. $$
In order to show $T_{n;\alpha,\beta}(\bm{X}) \to \infty$ as $n\to\infty$, we will show that $A_n\to\infty$ with probability 1 and that $B_n$ cannot be arbitrary large negative. 

Part $A_n$ can be further decomposed as follows: 
\begin{equation*}
    \begin{split}
        A_n&\ge a_{n;\alpha}c_0 n^{-1} \max_{i\in S} F^{-1}[1-p(X_i)|\alpha, \beta] + a_{n;\alpha} (\sum_{j \in S} w_j -c_0n^{-1})\min_{i\in S} F^{-1}[1-p(X_i)|\alpha, \beta] \\
        & := A_{n,1}+A_{n,2}
    \end{split}
\end{equation*}
In the following arguments, we will prove that $A_n\to\infty$ with probability 1 by showing that $A_{n,1}$ can be arbitrarily large whereas $A_{n,2}>o_p(1)$ as $n\to\infty$.

Since $F^{-1}[1-p(x)|\alpha,\beta]$ is increasing in $x$, $A_{n, 1} = a_{n;\alpha}  c_0 n^{-1}   F^{-1}[1-p(\max_{i\in S}|X_i|)|\alpha, \beta]$.
Recall that the set of positive signals ($S_+$) is assumed to have cardinality no less than $|S|/2$. From Lemma 6 of \cite{cai2014lemma6} and using the same argument as in the proof of Theorem 3 of \cite{liu2020cauchy},
$\max_{i\in S}|X_i|\ge \mu_0 + \sqrt{2\log|S_{+}|} + o_p(1)$. Given the assumptions $\mu_0 = \sqrt{2r\log n}$ and $\sqrt{2\log|S_{+}|} \ge \sqrt{2(\gamma \log n -\log2)}$, we have  $ \max_{i\in S}(|X_i|)\to \infty$ with probability 1. 
Lemma \ref{lemma: large x} implies that,  as $n\to\infty$, 
\begin{equation*}
        \Pr\left\{ F^{-1}\left[1-p\left(\max_{i\in S}|X_i|\right)|\alpha, \beta\right] > g\left(\max_{i\in S}|X_i|\right) \right\} \to 1,
\end{equation*}
which is equivalent to
\begin{equation*}
    \Pr\left\{ A_{n,1} \ge  a_{n;\alpha} {c_0}{n^{-1}} c_{\alpha, \beta} \left(\max_{i\in S}|X_i|\right)^{1/\alpha} \exp\left[ \frac{\left(\max_{i\in S}|X_i|\right)^2}{2\alpha}\right]   \right\}\to 1.
\end{equation*}
Noting that $\max_{i\in S}|X_i|\ge \sqrt{2r\log n} + \sqrt{2\log|S_{+}|} + o_p(1)$, $\Pr\{\max_{i\in S}X_i>1\}\to1$  
and $\sqrt{\log|S_{+}|} \ge \sqrt{2(\gamma \log n -\log2)} \approx \sqrt{2\gamma \log n }$, we have
\begin{equation*}
\Pr\left\{ A_{n,1}
        \ge  a_{n;\alpha}  {c_0}{n^{-1}} c_{\alpha, \beta} \left[n^{(\sqrt{\gamma} + \sqrt{r})^2} \right]^{1/\alpha} \right\}\to 1. 
\end{equation*}
From Lemma \ref{lemma:an}, $ a_{n;\alpha} = \left(\sum_{j = 1}^n w_j^\alpha\right)^{-1/\alpha} \ge \min\{ n^{(\alpha-1)/\alpha},1\}$, therefore, as $n\to\infty$,
$$\Pr\left\{ A_{n,1}
        \ge {c_0}c_{\alpha, \beta}  \left[n^{(\sqrt{\gamma} + \sqrt{r})^2/\alpha -1+\min\{1-1/\alpha,0\} } \right] \right\}\to 1.$$
By Pert \ref{assump part: magnitude} of Assumption \ref{assump: alternative}, $\sqrt{r} + \sqrt{\gamma} > \max\{\sqrt{\alpha}, 1\}$, we have $n^{(\sqrt{\gamma} + \sqrt{r})^2/\alpha -1/\alpha-1 } \to\infty$ as $n\to\infty$.
Therefore, we obtain that $A_{n, 1} \to \infty$ with probability tending to 1 as $n\to\infty$.

Next consider the part $A_{n, 2} =  a_{n;\alpha} \left(\sum_{j\in S}w_j-c_0n^{-1}\right)\min_{i\in S} F^{-1}[1-p(X_i)]$.  Suppose $\mu_1 = \mu_0$ without loss of generality, thus $X_1 = \mu_0 + Z_1$, where $Z_1\sim N(0, 1)$.
Let $\epsilon_n = n^{\alpha\gamma_0-1}$ with $\gamma<\gamma_0<\frac{1-\gamma}{\alpha}$.
Similarly to the proof of \cite{liu2020cauchy}, $\min_{i\in S} |X_i|$ is greater than any $\epsilon_n$ with probability 1 as $n\to\infty$ because
\begin{equation}\label{equation: minX}
\begin{split}
     \Pr\Big(\min_{i\in S} |X_i| <\epsilon_n \Big) &\le \sum_{i\in S}\Pr\Big(|X_i| <\epsilon_n \Big)\\
   & = n^{\gamma} \Pr\Big(|X_i| <\epsilon_n \Big)\\
   & = n^{\gamma} \left[ \Phi(\mu_0 + \epsilon_n) -\Phi(\mu_0 - \epsilon_n) \right]\\
   &<  n^{\gamma} \left[ 2\epsilon_n \phi(\mu_0-\varepsilon_n)  \right]\\
   &< n^{\gamma} \epsilon_n = n^{\gamma+ \alpha\gamma_0-1} = o(1). 
\end{split}
\end{equation}
Apply the increasing function $F^{-1}\left[1-p(x)|\alpha, \beta\right]$ on both $\min_{i\in S} |X_i|$ and $\epsilon_n$, 
equation (\ref{equation: minX}) is then equivalent to the statement that $F^{-1}[1-p(\min_{i\in S} |X_i|)|\alpha, \beta]$ is greater than any $ F^{-1}[1-p(\epsilon_n )|\alpha, \beta]$ with probability 1 as $n\to\infty$. 
Since $\epsilon_n \to 0$ as $n\to \infty$, we can apply Lemma \ref{lemma: small x} to find the lower bound of $ F^{-1}[1-p(\epsilon_n )|\alpha, \beta]$, which implies the bound of $A_{n, 2}$ as follows:
$$\Pr\left\{ |A_{n,2}| >a_{n;\alpha} (n^{\gamma-1}-c_0n^{-1}) |\tilde{g}(\epsilon_n )| \right\}\to 1.$$
With the assumption that there is a constant $c_0$ such that $\min_{i=1}^n w_i \ge c_0/n$, we have $a_n< n^{1-1/\alpha}c_0^{-1}$. As $n\to\infty$, $e^{\epsilon_n^2/(2\alpha)} \to 1$, and thus 
\begin{equation*}
    \begin{split}
      a_n (n^{\gamma-1} -c_0n^{-1})|\tilde{g}(\varepsilon_n)|
      &< 
        \tilde{c}_{\alpha, \beta} a_n n^{\gamma-1}   \varepsilon_n^{-1/\alpha}e^{\epsilon_n^2/(2\alpha)} \\
       & \le \tilde{c}_{\alpha, \beta} c_0^{-1}  n^{\gamma -\gamma_0 } \\
       & =o(1).
    \end{split}
\end{equation*}
Therefore, $A_{n,2}>o_p(1)$, which completes the proof of the statement that $A_n \to\infty$ with probability 1 as $n\to\infty$.

Next, we show $B_n$ cannot be arbitrary large negative. Under Part \ref{assump part: S complementary} of Assumption \ref{assump: alternative}, Theorem \ref{theorem: size} implies that as $n\to\infty$, $$\left(\frac{\sum_{j=1}^n w_j^\alpha}{\sum_{k\in S^c}w_k^\alpha}\right)^{1/\alpha}B_n = \left(\frac{\sum_{j=1}^n w_j^\alpha}{\sum_{k\in S^c}w_k^\alpha}\right)^{1/\alpha}a_{n;\alpha}\sum_{i\in S^c}w_i F^{-1}\left(1-p_i|\alpha, \beta\right) \overset{d}{\to} W_{0;\alpha, \beta},$$ 
where $W_{0;\alpha, \beta}$ follows $\bm{S}(\alpha, \beta)$.

Let $\delta_{\epsilon_n} = \left[\frac{1 - \beta}{\pi\epsilon_n}\Gamma(\alpha)\sin\left(\frac{\pi\alpha}{2}\right) \left(\frac{\sum_{k\in S^c}w_k^\alpha }{ \sum_{j=1}^nw_j^\alpha}\right) \right]^{1/\alpha}$, where $\epsilon_n = n^{\alpha\gamma_0-1}$ with $\gamma<\gamma_0<\frac{1-\gamma}{\alpha}$. Notice that as $n\to\infty$, $\epsilon_n\to 0$, and $\delta_{\epsilon_n}\left(\frac{\sum_{j=1}^nw_j^\alpha }{\sum_{k\in S^c}w_k^\alpha}\right)^{1/\alpha} = \left[\frac{1-\beta}{\pi\epsilon_n}\Gamma(\alpha)\sin\left(\frac{\pi\alpha}{2}\right)\right]^{1/\alpha} \to \infty$. We first show $-1<\beta<1$ case. According to the tail approximation of Theorem 1.2 of \cite{nolan2020modeling}, when $0<\alpha<2$ and $-1\le \beta< 1$,
\begin{equation}\label{equation: B_n bounded}
    \begin{split}
        \Pr\left(B_n<- \delta_{\epsilon_n}\right) &\sim  \Pr\left[ W_{0;\alpha, \beta} <- \delta_{\epsilon_n}\left(\frac{\sum_{j=1}^nw_j^\alpha }{\sum_{k\in S^c}w_k^\alpha}\right)^{1/\alpha}\right]\\
        & \sim \frac{1-\beta}{\pi}\Gamma(\alpha)\sin\left(\frac{\pi\alpha}{2}\right) \delta_{\epsilon_n}^{-\alpha} \left(\frac{\sum_{j=1}^nw_j^\alpha }{\sum_{k\in S^c}w_k^\alpha}\right)^{-1}\\
        &=\varepsilon_n,
    \end{split}
\end{equation}
for large enough $n$. 
Equation (\ref{equation: B_n bounded}) implies that for any $\epsilon_n\to 0$, there exist an $\delta>\delta_{\epsilon_n}$ such that $\Pr\left(B_n<- \delta\right)< \epsilon_n$ as $n\to\infty$. 

When $\beta=1$, the distribution is totally skewed to the right, and consequently, for all $i \in S^c$, $\Pr\left(W_{i;\alpha, 1} < -\delta_{\epsilon_n}\right) < \Pr\left(W_{i;\alpha, \beta} < -\delta_{\epsilon_n}\right) $ for any $\beta<1$. Therefore, equation (\ref{equation: B_n bounded}) holds for $\beta=1$ as well. That is, $B_n$ cannot be arbitrary large negative for all $0<\alpha<2$ and $-1<\beta\le1$, which finishes the proof.
\end{proof}

\end{appendices}

\bibliographystyle{chicago}
\bibliography{ref}

\end{document}